\documentclass{article}
\usepackage[utf8]{inputenc}
\usepackage[english]{babel}
\usepackage{amsmath,amsthm,amssymb}
\usepackage{qcircuit}
\usepackage{physics}
\usepackage{caption}
\usepackage{subcaption}
\usepackage[ruled]{algorithm2e}
\usepackage{appendix}
\usepackage{floatrow}
\usepackage{arxiv}
\usepackage{graphicx}

\newcommand{\B}{\mathbb{B}}
\newcommand{\N}{\mathbb{N}}
\newcommand{\C}{\mathbb{C}}
\newcommand{\BB}[1][n]{\B^{\B^{#1}}}

\newfloatcommand{capbtabbox}{table}[][\FBwidth]
\newtheorem{thrm}{Theorem}[section]
\newtheorem{lemma}{Lemma}[section]
\newtheorem{corlry}{Corollary}[section]
\theoremstyle{definition}
\newtheorem{ex}{Example}[section]

\title{Tunable Quantum Neural Networks for Boolean Functions}

\author{
  Viet Pham Ngoc\\
    Imperial College London\\ London, United Kingdom\\
    \texttt{viet.pham-ngoc17@imperial.ac.uk}
    \And
    Herbert Wiklicky\\
    Imperial College London\\ London, United Kingdom\\
    \texttt{h.wiklicky@imperial.ac.uk}
}

\begin{document}
\maketitle

\begin{abstract}
In this paper we propose a new approach to quantum neural networks. Our multi-layer architecture does without the measurements usually emulating the non-linear activation functions that are characteristic of the classical neural networks. Despite this, our proposed architecture is still able to learn any Boolean function. This ability arises from the correspondence that exists between a Boolean function and a particular quantum circuit made of multi-controlled $\mathbf{X}$ gates. This correspondence is built via a polynomial representation of the function called the algebraic normal form. We use this construction to introduce the idea of a generic quantum circuit whose gates can be tuned to learn any Boolean functions. In order to perform the learning task, we have devised an algorithm that leverages the absence of measurements. When presented with a superposition of all the binary inputs of length $n$, the network can learn the target function in at most $n+1$ updates.
\end{abstract}

\section{Introduction}
In the past years, architectures based on artificial neural networks have been successfully applied to numerous fields. Nonetheless, these architectures often come with a significant complexity, particularly during the training phase, so much so that alleviating this complexity has become an active area of research. The main propositions stemming from this area are either alternative algorithms \cite{Li2017} or dedicated hardware to accelerate the computations \cite{Zhao2016}.

In the meantime, quantum versions of some other machine learning techniques, for supervised and unsupervised learning tasks alike, have been devised and have shown quantum speed up \cite{Lloyd2013}. These developments seem to indicate that a quantum version of artificial neural networks could exhibit similar speed up. To this end, several architectures have been proposed to implement quantum artificial neural networks \cite{Schuld2014}. These propositions mimic to different degrees the functioning of the classical versions while taking advantage of quantum properties such as superposition to exponentially reduce the dimension of the state space \cite{Rebentrost2018,Tacchino2019}. One persistent hurdle arising from this approach is the implementation of the non-linear activation function. Most of the proposed solutions to overcome this difficulty rely on the use of measurements to provide the non-linear behaviour at the cost of losing the quantum nature of the state. 

Here we propose an alternative architecture devoid of measurement, thus preserving the quantum information, but that is still able to learn any Boolean function. This construction relies on the correspondence, introduced by \cite{Younes2004} and most recently by \cite{Bogdanov2019}, between the algebraic normal form of a Boolean function and a quantum circuit made of multi-controlled $\mathbf{X}$ gates acting on a single ancillary qubit. We provide a proof that this correspondence is correct and unique by building a group isomorphism between the set of Boolean functions and the set of the multi-controlled $\mathbf{X}$ gates. This allows us to introduce a general quantum circuit whose gates can be tuned in order to exactly learn any Boolean functions.

This circuit can be likened to a neural network in the sense that each gate can only perform a simple operation on the input data as well as the output of the previous gate. The circuit taken as a whole is then able to compute complex Boolean functions. In order to train this particular network we have devised a learning algorithm. While \cite{Younes2003} have introduced an algorithm to automatically build circuits similar to the ones presented in this paper, ours does not need to have access to the truth table of the target function and takes an approach that is more in the fashion of the training algorithms used with classical neural networks. Furthermore it takes advantage of the lack of measurements to reduce the number of updates needed until convergence. We show that provided with a superposition of all the inputs of length $n$, our network can exactly learn the target function within $n+1$ updates. In order to validate our theoretical results, we have implemented our approach on the IBM quantum simulator for the cases $n=2$ and $n=3$.

\section{The Algebraic Normal Form of Boolean Functions}\label{sec:ANF}
Let $\B=\{0,1\}$ the set of the Boolean values and $\B^{\B^n}$ the set of functions from $\B^n$ to $\B$. \newline For $u=u_0 \ldots u_{n-1} \in \B^n$ we note:
\begin{equation}\label{eq:1u}
    1_u = \left \{i \in \{0, \ldots , n-1\} \ | \ u_i = 1\right \}
\end{equation}
And for $x=x_0 \ldots x_{n-1} \in \B^n$, $x^u$ is defined by 
$$
    x^u = \prod_{i \in 1_u}{x_i}
$$
Now for $u \in \B^n$, we introduce $m_u \colon \B^n \to \B$ with $m_u \colon x \mapsto x^u$. 

Let $f \in \B^{\B^n}$, then $f$ has a unique polynomial representation of the form:
\begin{equation}\label{eq:ANF}
    f = \bigoplus_{u \in \B^n}{c^f_u m_u}
\end{equation}
The notations introduced in (\ref{eq:ANF}) are as follow. The binary operator $\oplus$ represents the logical operator $\mathbf{XOR}$. The terms $m_u$ are called monomials and the coefficients $c^f_u \in \B$ indicate the presence or the absence of the corresponding monomial.    

This representation is called the algebraic normal form (ANF). We are interested in the ANF as it makes explicit the relation between the inputs and the outputs of a Boolean function by using two simple Boolean operations: the $\mathbf{XOR}$ and the $\mathbf{AND}$. While there exist other polynomial representations, this particular form allows for an easy translation of a Boolean function into a quantum circuit as will be shown below. In Table \ref{tab:ANF(f)} we have gathered the algebraic normal form of some well-known functions.
\begin{table}[h]
    \centering
    \begin{tabular}{|c||c|c|c|c|}
        \hline
        f      & $\mathbf{0}$ & $\mathbf{1}$ & $\mathbf{NOT}(x_0)$ & $\mathbf{OR}(x_0x_1)$           \\
        \hline
        ANF(f) & 0            & 1            & $1 \oplus x_0$      & $x_0 \oplus x_1 \oplus x_0.x_1$ \\
        \hline
    \end{tabular}
    \caption{ANF of some Boolean functions}
    \label{tab:ANF(f)}
\end{table}
\\
We now show that the set of Boolean functions possesses a group structure when fitted with the operator $\mathbf{XOR}$ and that a Boolean function has a unique algebraic normal form. 
\begin{lemma}\label{lem:B^B^n}
Let $f,g \in \B^{\B^n}$, we define the operation $f \oplus g$ as $f \oplus g \colon x \mapsto f(x) \oplus g(x)$, then:
\\
$\left(\B^{\B^n},\oplus\right)$ is a finite commutative group where the identity element is the constant function $\mathbf{0}$ and each element is its own inverse.
\end{lemma}
\begin{proof}
    These results stem from the properties of the operator $\mathbf{XOR}$.
\end{proof}
Let $\mathcal{M} = \{m_u \ | \ u \in \B^n\}$ be the set of the monomials over $\B^n$, then we have the following:  
\begin{thrm}\label{th:A}
    Let $\mathcal{A}$ be the subgroup of $\left(\B^{\B^n},\oplus\right)$ generated by $\mathcal{M}$, then:
    $$
        \mathcal{A} = \B^{\B^n}
    $$
\end{thrm}
\begin{proof}
    The proof can be found in Appendix \ref{ap:proof}
\end{proof}
Theorem \ref{th:A} shows the existence as well as the uniqueness of the algebraic normal form of a Boolean function. It is therefore equivalent to consider a Boolean functions or its ANF. 

For this equivalence to be complete, we now present a way to construct the ANF of a function called the Method of Indeterminate Coefficients \cite{Kurgalin2018}. For the sake of clarity we use a simple example but the generalization is straight-forward. 
\begin{ex}
    Let $f \in \B^{\B^2}$ defined by Table \ref{tab:f_ex}.
    \begin{table}[h]
        \centering
        \begin{tabular}{|c||c|c|c|c|}
            \hline
            x & 00 & 01 & 10 & 11 \\
            \hline
            f(x) & 1 & 0 & 1 & 1 \\
            \hline
        \end{tabular}
        \caption{Truth table of $f$}
        \label{tab:f_ex}
    \end{table}
    \\
    These values can be gathered in the form of a vector $v_f = (1,0,1,1)^t$ and we are looking for the vector $c=\left(c_{00}^f,c_{01}^f,c_{10}^f,c_{11}^f\right)^t$ verifying $\mathbf{T}c=v_f$ where $\mathbf{T}$ is the matrix such that $\mathbf{T}_{x,u}=m_u(x)$. In our case:
    $$
        \mathbf{T}=\begin{pmatrix}
            1 & 0 & 0 & 0 \\
            1 & 1 & 0 & 0 \\
            1 & 0 & 1 & 0 \\
            1 & 1 & 1 & 1
        \end{pmatrix}
    $$
    It can be shown that $\mathbf{T} = \mathbf{T}^{-1}$ in $\B^{2 \times 2}$, hence $c = \mathbf{T}v_f$ and $c=(1,1,0,1)^t$. This leads to $f = m_{00} \oplus m_{01} \oplus m_{11}$ or $f(x_0x_1)= 1 \oplus x_1 \oplus x_0.x_1$
\end{ex}

\section{From Algebraic Normal Form to Quantum Circuit}\label{sec:A->Aq}
Using the ANF of a Boolean function, we show that we can easily design a quantum circuit expressing this function. We consider a quantum circuit operating on $n+1$ qubits $\ket{x_0 \ldots x_{n-1}}\ket{q_r}$ where the last qubit is the read-out qubit and we are interested in a particular set of gates. Let $u \in \B^n$ and define $\mathbf{C}_u$ as the multi-controlled $\mathbf{X}$ gate that is acting on the read-out qubit and controlled by the qubits $\{\ket{x_i} \ | \ i \in 1_u\}$ where $1_u$ is defined as in (\ref{eq:1u}). Figure \ref{fig:C(2)} represents the 4 possible gates for $n=2$.
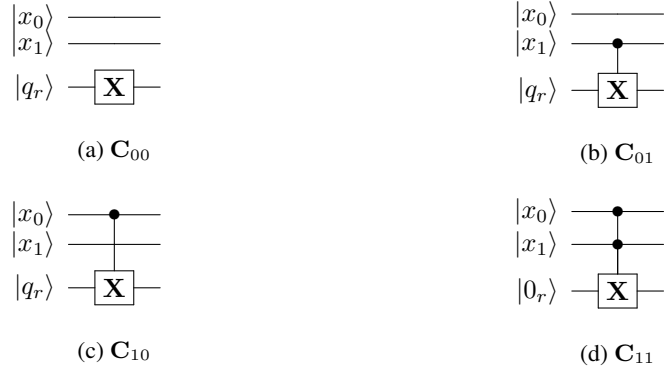
\begin{figure}[h]
    \centering
    \begin{subfigure}{0.4\textwidth}
        \centering
        $$
        \Qcircuit @R=1em @C=1em{
            \lstick{\ket{x_0}} & \qw      & \qw \\
            \lstick{\ket{x_1}} & \qw      & \qw \\
            \lstick{\ket{q_r}} & \gate{\mathbf{X}} & \qw
        }
        $$
        \caption{$\mathbf{C}_{00}$}
        \label{fig:C_00}
    \end{subfigure}
    \begin{subfigure}{0.4\textwidth}
        \centering
        $$
        \Qcircuit @R=1em @C=1em{
            \lstick{\ket{x_0}} & \qw      & \qw \\
            \lstick{\ket{x_1}} & \ctrl{1} & \qw \\
            \lstick{\ket{q_r}} & \gate{\mathbf{X}} & \qw
        }
        $$
        \caption{$\mathbf{C}_{01}$}
        \label{fig:C_{01}}
    \end{subfigure}
    \begin{subfigure}{0.4\textwidth}
        \centering
        $$
        \Qcircuit @R=1em @C=1em{
            \lstick{\ket{x_0}} & \ctrl{2} & \qw \\
            \lstick{\ket{x_1}} & \qw      & \qw \\
            \lstick{\ket{q_r}} & \gate{\mathbf{X}} & \qw
        }
        $$
        \caption{$\mathbf{C}_{10}$}
        \label{fig:C_{10}}
    \end{subfigure}
    \begin{subfigure}{0.4\textwidth}
        \centering
        $$
        \Qcircuit @R=1em @C=1em{
            \lstick{\ket{x_0}} & \ctrl{2} & \qw \\
            \lstick{\ket{x_1}} & \ctrl{1} & \qw \\
            \lstick{\ket{0_r}} & \gate{\mathbf{X}} & \qw
        }
        $$
        \caption{$\mathbf{C}_{11}$}
        \label{fig:C_{11}}
    \end{subfigure}
    \caption{The set of the possible controlled $\mathbf{X}$ gates for $n=2$}
    \label{fig:C(2)}
\end{figure}
\\
We now present some intermediary results that exhibit the correspondence existing between the algebraic normal form of a Boolean function and the quantum circuit able to compute this function over the ancillary qubit. 
\begin{lemma}\label{lem:C_u}
Let $u \in \B^n$, then:
$$
    \forall x \in \B^{n}, \forall q_r \in \B, \mathbf{C}_u\ket{x}\ket{q_r} = \ket{x}\ket{q_r \oplus m_u(x)}
$$
\end{lemma}

\begin{proof}
The proof can be found in Appendix \ref{ap:proof}
\end{proof}

\begin{lemma}\label{lem:A_Q}
Let $\mathcal{C} = \left\{\mathbf{C}_u \ | \ u \in \B^n\right\}$ and $\mathcal{A}_Q$ the subgroup of the unitary group $U(n+1)$ that is generated by $\mathcal{C}$, then:
\\
$\mathcal{A}_Q$ is a finite commutative group whose identity element is the Identity matrix and where each element is its own inverse. 
\end{lemma}
\begin{proof}
    The proof can be found in Appendix \ref{ap:proof}
\end{proof}
We now introduce the group morphism $\Phi \colon \mathcal{A} \to \mathcal{A}_Q$ such that:
$$
    \forall u \in \B^n, \Phi(m_u) = \mathbf{C}_{u}
$$
This morphism leads to the following Lemma:
\begin{lemma}\label{lem:phi(f)}
Let $f \in \B^{\B^n}$ then:
$$
    \forall x \in \B^{n}, \forall q_r \in \B, \ \Phi(f)\ket{x}\ket{q_r} = \ket{x}\ket{q_r \oplus f(x)}
$$
\end{lemma}
\begin{proof}
    It suffices to take the ANF of $f$ and apply $\Phi$. Lemma \ref{lem:C_u} then allows us to conclude.
\end{proof}
Thanks to these lemmas, we can now prove that $\mathcal{A}$ and $\mathcal{A}_Q$ are equivalent.
\begin{thrm}\label{th:A->A_Q}
Let $\Phi$ be the group morphism such that $\Phi(m_u)=\mathbf{C}_u$. Then:
\\
$\Phi$ is an isomorphism from $\mathcal{A}$ to $\mathcal{A}_Q$.
\end{thrm}
\begin{proof}
The proof can be found in Appendix \ref{ap:proof}
\end{proof}

Theorem \ref{th:A->A_Q} ensures that, given a Boolean function, we can design a quantum circuit able to express this function. To do so, it suffices to transform each monomial of the ANF into its corresponding gate. The circuit is then built by multiplying these gates. The commutativity of $\mathcal{A}_Q$ ensures that a change in the order of the gates leaves the resulting circuit invariant. More over, it shows that there is a reciprocal correspondence between a Boolean function and the quantum circuit computing this function. 

\begin{ex}\label{ex:f_example}
We illustrate the process of constructing such a quantum circuit with the function $f$  whose values are provided in Table \ref{tab:f}. 
\begin{table}[h]
    \centering
    \begin{tabular}{|c|c|c|c|}
        \hline
        $x$ & $f(x)$ & $x$ & $f(x)$ \\
        \hline
        000 & 0 & 100 & 1 \\
        \hline
        001 & 0 & 101 & 0 \\
        \hline
        010 & 1 & 110 & 0 \\
        \hline
        011 & 0 & 111 & 1 \\
        \hline
    \end{tabular}
    \caption{Values of $f$}
    \label{tab:f}
\end{table}
\\
Applying the Method of Indeterminate Coefficients yields the algebraic normal form of $f$:
$$
    f(x_0x_1x_2) = x_1 \oplus x_0 \oplus x_1.x_2 \oplus x_0.x_2 \oplus x_0.x_1.x_2
$$
Now applying the isomorphism $\Phi$ to $f$, we get :
\begin{align*}
    \Phi(f) &= \mathbf{C}_{010}\mathbf{C}_{100}\mathbf{C}_{011}\mathbf{C}_{101}\mathbf{C}_{111} \\ 
            &= \mathbf{C}_{111}\mathbf{C}_{101}\mathbf{C}_{100}\mathbf{C}_{011}\mathbf{C}_{010}
\end{align*}
The resulting circuit is depicted in Figure \ref{fig:C_f}.
\begin{figure}[h]
    \centering
    $$
    \Qcircuit @R=1em @C=1em{
        \lstick{\ket{x_0}} & \qw      & \qw      & \ctrl{3} & \ctrl{3} & \ctrl{3} & \qw \\
        \lstick{\ket{x_1}} & \ctrl{2} & \ctrl{2} & \qw      & \qw      & \ctrl{2} & \qw \\
        \lstick{\ket{x_2}} & \qw      & \ctrl{1} & \qw      & \ctrl{1} & \ctrl{1} & \qw \\
        \lstick{\ket{q_r}} & \gate{\mathbf{X}} & \gate{\mathbf{X}} & \gate{\mathbf{X}} & \gate{\mathbf{X}} & \gate{\mathbf{X}} & \qw 
    }
    $$
    \caption{Quantum circuit representing $\Phi(f)$} 
    \label{fig:C_f}
\end{figure}
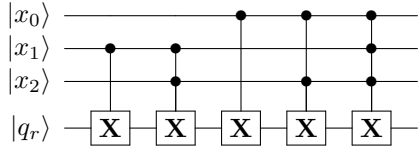
\\
We can check that we do have $\Phi(f)\ket{x}\ket{0}=\ket{x}\ket{f(x)}$ for $x \in \B^3$. For example: 
$$
    \Phi(f)\ket{010}\ket{0}=\ket{010}\ket{1}=\ket{010}\ket{f(010)}
$$ 
And 
$$
    \Phi(f)\ket{111}\ket{0}=\ket{111}\ket{1}=\ket{111}\ket{f(111)}
$$ 
So the circuit $\Phi(f)$ is able to compute $f$ which is what we intended to achieve.
\end{ex}
These preliminary results lead to the concept of a generic quantum circuit. This circuit is made of gates whose action on the read-out qubit can be tuned so that the resulting circuit is able to express any given Boolean function.

\section{Tunable Quantum Neural Network}
As previously, we work with a quantum circuit operating on n+1 qubits $\ket{x_o \ldots x_{n-1}}\ket{q_r}$. We recall that for $u \in \B^n$, $\mathbf{C}_u$ is the multi-controlled $X$ gate, controlled by the qubits $\{\ket{x_i} \ | \ i \in 1_u\}$ and $1_u$ is defined as in (\ref{eq:1u}). We introduce $\mathbf{G}_u$ the tunable quantum gate whose value can either be $\mathbf{I}$, the identity gate, or $\mathbf{C}_u$, that is $\mathbf{G}_u \in \{\mathbf{I}, \mathbf{C}_u\}$. $\mathbf{G}_u$ can be seen as the quantum version of the neuron as it performs simple local computations and can be parameterised. Let $\mathbf{TNN}$ be the tunable quantum neural network, $\mathbf{TNN}$ is defined by:
$$
    \mathbf{TNN} = \prod_{u \in \B^n}{\mathbf{G}_u}
$$
We have shown with Lemma \ref{lem:A_Q} that the gates $\mathbf{C}_u$ commute together and $\mathbf{I}$ commuting with all matrices, it comes that the order in which the product is done does not change the overall circuit. Tunable neural networks for the cases $n=2$ and $n=3$ are pictured in Figure \ref{fig:TunableQC}. 

For $n \in \N$ such a circuit contains $2^n$ gates. Each gate having two possible values, there exists in total $2^{2^n}$ different circuits, meaning that the set of all the tuned circuit is $\mathcal{A}_Q$ and thus any Boolean function can be expressed by a properly tuned network.
\begin{figure}[h]
    \begin{subfigure}{\textwidth}
        \centering
        $$
        \Qcircuit @R=1em @C=1em{
            \lstick{\ket{x_0}} & \qw        & \qw        & \ctrl{2}   & \ctrl{2}      & \qw \\
            \lstick{\ket{x_1}} & \qw        & \ctrl{1}   & \qw        & \ctrl{1}      & \qw \\
            \lstick{\ket{q_r}} & \gate{\mathbf{G}_{00}} & \gate{\mathbf{G}_{01}} & \gate{\mathbf{G}_{10}} & \gate{\mathbf{G}_{11}}  & \qw 
        }
        $$
        \caption{Tunable neural network with 2 qubits input}
        \label{fig:TunableQC2}
    \end{subfigure}
    \begin{subfigure}{\textwidth}
    \tiny
    \centering
        $$
        \Qcircuit @R=1em @C=0.5em{
            \lstick{\ket{x_0}} & \qw        & \qw        & \qw        & \qw        & \ctrl{3}     & \ctrl{3}     & \ctrl{3}     & \ctrl{3}     & \qw    \\
            \lstick{\ket{x_1}} & \qw        & \qw        & \ctrl{2}   & \ctrl{2}   & \qw          & \qw          & \ctrl{2}     & \ctrl{2}     & \qw    \\
            \lstick{\ket{x_2}} & \qw        & \ctrl{1}   & \qw        & \ctrl{1}   & \qw          & \ctrl{1}     & \qw          & \ctrl{1}     & \qw    \\
            \lstick{\ket{q_r}} & \gate{\mathbf{G}_{000}} & \gate{\mathbf{G}_{001}} & \gate{\mathbf{G}_{010}} & \gate{\mathbf{G}_{011}} & \gate{\mathbf{G}_{100}} & \gate{\mathbf{G}_{101}} & \gate{\mathbf{G}_{110}} & \gate{\mathbf{G}_{111}} & \qw 
        }
        $$
        \caption{Tunable neural network with 3 qubits input}
        \label{fig:TunableQC3}
    \end{subfigure}
    \caption{Tunable neural networks for different input size}
    \label{fig:TunableQC}
\end{figure}
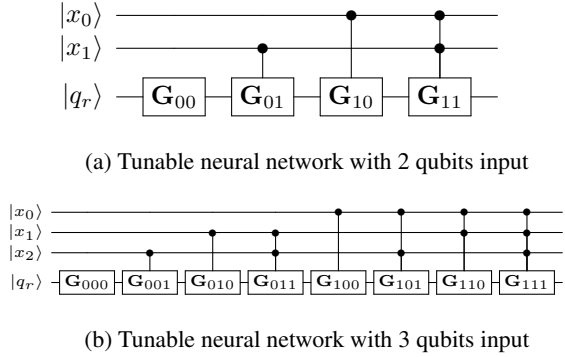

\section{Learning Algorithm}\label{sec:alg}
Given such a circuit and a Boolean function, $f \in \B^{\B^n}$, we introduce a learning algorithm resulting in a correctly tuned neural network that is able to express $f$. We first outline a general version of the algorithm and then discuss the details further. This algorithm intends to use quantum superposition in order to reduce the number of updates performed during the learning phase. Let $\ket{\Psi}$ be such superposition:
$$
    \ket{\Psi} = \sum_{x \in \B^n}{a_x\ket{x}\ket{0}}
$$
Where the last qubit is the read-out qubit. 

Suppose further that given a superposition, we are able to identify some of the states present in this superposition. More precisely, given a network $\mathbf{TNN}$ and $\ket{\Psi'} = \mathbf{TNN}\ket{\Psi}$, we are able to identify all the states of the form $\ket{x}\ket{1 \oplus f(x)}$ in $\ket{\Psi'}$. Let us call this operation $qt(f)\left(\ket{\Psi'}\right)$, then 
\begin{align*}
    &qt(f)\left(\ket{\Psi'}\right) = \\
    &\{x \in \B^n \ | \ a_x \neq 0, \ \mathbf{TNN}\ket{x}\ket{0} = \ket{x}\ket{1 \oplus f(x)}\}
\end{align*}
The goal of the operator $qt$ is thus to identify the inputs for which the output by the circuit is different from that of the target function. 

We start with $\mathbf{TNN}^{(0)}$ the circuit where all the gates are initialised with $\mathbf{I}$. For $k \in \N$, $\mathbf{TNN}^{(k)}$ is the circuit resulting from the $k$-th update and we denote by $E^{(k)} = qt(f)\left(\mathbf{TNN}^{(k)}\ket{\Psi}\right)$ the set of the inputs for which the output by $\mathbf{TNN}^{(k)}$ is erroneous. For example $E^{(0)} = f^{-1}(\{1\})$. Given these notations, we define the update rule as follow: 
\begin{itemize}
    \item Determine $E^{(k)}$
    \item For $u \in E^{(k)}$ switch the value of the corresponding gate $\mathbf{G}_u$, resulting in $\mathbf{TNN}^{(k+1)}$
\end{itemize}
This last step can be realized by multiplying $\mathbf{G}_u$ by $\mathbf{C}_u$. Indeed, if $\mathbf{G}_u = \mathbf{I}$ then $\mathbf{G}_u\mathbf{C}_u$ = $\mathbf{C}_u$, else we have $\mathbf{G}_u = \mathbf{C}_u$ and $\mathbf{G}_u\mathbf{C}_u$ = $\mathbf{I}$
The algorithm terminates when we reach an updated circuit $\mathbf{TNN}^{(h)}$ such that $E^{(h)} = \varnothing$. The learning algorithm can be summarised as shown in the Algorithm \ref{alg:tuning}
\begin{algorithm}[h]
    \SetAlgoLined
    $\mathbf{TNN} \leftarrow \mathbf{I}$\;
    $E \leftarrow qt(f)(\ket{\Psi})$\;
    \While{$E \neq \varnothing$}{
        \For{$u \in E$}{
            $\mathbf{G}_u \leftarrow \mathbf{G}_u\mathbf{C}_u$\;
        }
        $\mathbf{TNN} \leftarrow \prod_{u \in \B^n}{\mathbf{G}_u}$\;
        $\ket{\Psi '} \leftarrow \mathbf{TNN}\ket{\Psi}$\;
        $E \leftarrow qt(f)\left(\ket{\Psi'}\right)$\;
    }
    \caption{Learning algorithm}
    \label{alg:tuning}
\end{algorithm}

\subsection{Example}
Let us run this algorithm on an example. We want to tune the circuit in order for it to express the function $f$ introduced in Example \ref{ex:f_example}, the values of which are gathered in Table \ref{tab:f}. Let $\ket{\Psi}$ be the superposition we are working with:
$$
    \ket{\Psi} = \sum_{x \in \B^3}{a_x\ket{x}\ket{0}}
$$
\\
Where $a_x \neq 0$ for $x \in \B^3$. We start with the circuit $\mathbf{TNN}^{(0)} = \mathbf{I}$ and 
\begin{align*}
    \ket{\Psi^{(0 )}} &= \mathbf{TNN}^{(0)}\ket{\Psi} \\
                      &= \sum_{\substack{x \in \B^3\\f(x)=0}}{a_x\ket{x}\ket{0}}+\sum_{\substack{x \in \B^3\\f(x)=1}}{a_x\ket{x}\ket{0}}
\end{align*}
Performing $qt(f)(\ket{\Psi^{(0)}})$ then yields: 
$$
    E^{(0)}=\{010,100,111\}
$$ 
We thus have to switch the value of the corresponding gates: $\mathbf{G}_{010}, \ \mathbf{G}_{100}$ and $\mathbf{G}_{111}$ resulting in $\mathbf{TNN}^{(1)}$ as depicted in Figure \ref{fig:TC1}.
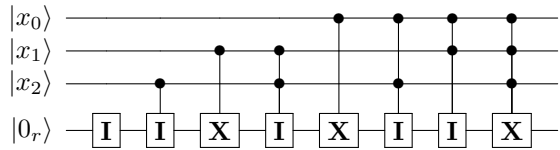
\begin{figure}[h]
    \centering
    $$
    \Qcircuit @R=1em @C=1em{
        \lstick{\ket{x_0}} & \qw        & \qw        & \qw        & \qw        & \ctrl{3}     & \ctrl{3}     & \ctrl{3}     & \ctrl{3}     & \qw \\
        \lstick{\ket{x_1}} & \qw        & \qw        & \ctrl{2}   & \ctrl{2}   & \qw          & \qw          & \ctrl{2}     & \ctrl{2}     & \qw \\
        \lstick{\ket{x_2}} & \qw        & \ctrl{1}   & \qw        & \ctrl{1}   & \qw          & \ctrl{1}     & \qw          & \ctrl{1}     & \qw \\
        \lstick{\ket{0_r}} & \gate{\mathbf{I}} & \gate{\mathbf{I}} & \gate{\mathbf{X}} & \gate{\mathbf{I}} & \gate{\mathbf{X}} & \gate{\mathbf{I}} & \gate{\mathbf{I}} & \gate{\mathbf{X}} & \qw
    }
    $$
    \caption{$\mathbf{TNN}^{(1)}$, the circuit obtained after the first update}
    \label{fig:TC1}
\end{figure}
\\
Continuing the algorithm: $\ket{\Psi^{(1)}} = \mathbf{TNN}^{(1)}\ket{\Psi}$ with
\begin{align*}
    \ket{\Psi^{(1)}} &= \left(a_{000}\ket{000} + a_{001}\ket{001} + a_{110}\ket{110}\right)\ket{0} \\
    &+ \left(a_{010}\ket{010} + a_{100}\ket{100} + a_{111}\ket{111}\right)\ket{1} \\
    &+ (a_{011}\ket{011} + a_{101}\ket{101})\ket{1}
\end{align*}
\\
Once again we perform $qt(f)(\ket{\Psi^{(1)}})$ and get 
$$
    E^{(1)}=\{011,101\}
$$ 
Applying the update rule, we switch the gates $\mathbf{G}_{011}$ and $\mathbf{G}_{101}$, resulting in $\mathbf{TNN}^{(2)}$ as represented in Figure \ref{fig:TC2}. 
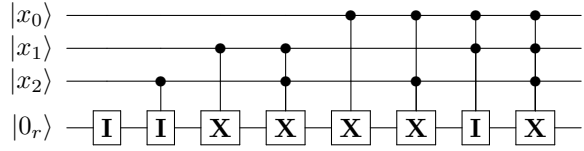
\begin{figure}[h]
    \centering
    $$
    \Qcircuit @R=1em @C=1em{
        \lstick{\ket{x_0}} & \qw        & \qw        & \qw        & \qw        & \ctrl{3}     & \ctrl{3}     & \ctrl{3}     & \ctrl{3}     & \qw \\
        \lstick{\ket{x_1}} & \qw        & \qw        & \ctrl{2}   & \ctrl{2}   & \qw          & \qw          & \ctrl{2}     & \ctrl{2}     & \qw \\
        \lstick{\ket{x_2}} & \qw        & \ctrl{1}   & \qw        & \ctrl{1}   & \qw          & \ctrl{1}    & \qw          & \ctrl{1}     & \qw \\
        \lstick{\ket{0_r}} & \gate{\mathbf{I}} & \gate{\mathbf{I}} & \gate{\mathbf{X}} & \gate{\mathbf{X}} & \gate{\mathbf{X}} & \gate{\mathbf{X}} & \gate{\mathbf{I}} & \gate{\mathbf{X}} & \qw
    }
    $$
    \caption{$\mathbf{TNN}^{(2)}$, the circuit obtained after the second update}
    \label{fig:TC2}
\end{figure}
\\
Applying $\mathbf{TNN}^{(2)}$ to $\ket{\Psi}$ now yields:
$$
    \ket{\Psi^{(2)}} = \sum_{\substack{x \in \B^3\\f(x)=0}}{a_x\ket{x}\ket{0}}+\sum_{\substack{x \in \B^3\\f(x)=1}}{a_x\ket{x}\ket{1}}
$$
And $E^{(2)}=qt(f)(\ket{\Psi^{(2)}})=\varnothing$ which is the termination condition of the algorithm. Applying $\Phi^{-1}$, the inverse of the isomorphism introduced in Section \ref{sec:A->Aq}, we get $\Phi^{-1}(\mathbf{TNN}^{(2)}) = \tilde{f}$ where
$$
    \tilde{f} \colon (x_0x_1x_2) \mapsto x_1 \oplus x_1.x_2 \oplus x_0 \oplus x_0.x_2 \oplus x_0.x_1.x_2 
$$
As this is the algebraic normal form of $f$ we conclude that $\tilde{f}=f$ and we have tuned the circuit properly.

\subsection{Proof of Termination and Correctness of the Algorithm}
We show here that provided with a superposition of all possible states and the operation $qt$, the algorithm will terminate after at most $n+1$ updates. We begin with the following:
\begin{lemma}\label{lem:T_x}
    Let $x \in \B^n$, we denote by $\mathcal{T}_x$ the set of gates that can be triggered by $\ket{x}\ket{q_r}$ and $w_H(x)=|1_x|$ the Hamming weight of $x$ then:
    $$
        \forall x \in \B^n, \ \mathcal{T}_x \subseteq \{\mathbf{G}_x\} \cup \{\mathbf{G}_u \ | \ w_H(u) < w_H(x)\}
    $$
\end{lemma}
\begin{proof}
    This comes from the fact that for $\mathbf{G}_u$ to be triggered by $x$, we must have $1_u \subseteq 1_x$ which leads to $ w_H(u) \leq w_H(x)$. Suppose now that $ w_H(u) = w_H(x)$. Because $|1_u|=w_H(u)$ and $|1_x|=w_H(x)$, necessarily we have $1_u = 1_x$ that is $u=x$.
\end{proof}
Lemma \ref{lem:T_x} shows that if $\mathbf{TNN} = \prod_{u \in \B^n}{\mathbf{G}_u}$ then for $x \in \B^n$ we have:
$$
    \mathbf{TNN}\ket{x}\ket{q_r} = \mathbf{G}_x\prod_{\substack{u \in \B^n\\w_H(u) < w_H(x)}}{\mathbf{G}_u}\ket{x}\ket{q_r}
$$
We will use this result to show the following:
\begin{lemma}\label{lem:k_update}
    Suppose that all the inputs are present in the superposition $\ket{\Psi}$, that is $a_x \neq 0$ for $x \in \B^n$. Then after the $k$-th update, the gates controlled by at most $k-1$ qubits will not be updated anymore.
\end{lemma}
\begin{proof}
    The proof can be found in Appendix \ref{ap:proof}
\end{proof}
From Lemma \ref{lem:k_update} stems the following corollary: 
\begin{corlry}\label{crlr:E^k}
    Let $E^{(k)}$ be the set of inputs for which the output by the quantum network is erroneous after the $k$-th update. Then:
    $$
        \forall k > 0, \ E^{(k)} \subseteq \{x \in \B^n \ | \ w_H(x) \geq k\}
    $$
\end{corlry}
\begin{proof}
    Let $k > 0$ and suppose there exist an $x \in \B^n$ such that $w_H(x) = s < k$ and $x \in E^{(k)}$. Then according to the update rule, the gate $\mathbf{G}_x$ will be updated which is in contradiction with Lemma \ref{lem:k_update}.
\end{proof}
We recall that the process terminates whenever a circuit $\mathbf{TNN}^{(h)}$ such that $E^{(h)} = \varnothing$ is reached. The previous results then allow us to prove:
\begin{thrm}\label{th:termination}
    When presented with a superposition of all possible inputs, the process will terminate after at most $n+1$ updates.
\end{thrm}
\begin{proof}
    This is a direct consequence of Corollary \ref{crlr:E^k}. Indeed, according to Corollary \ref{crlr:E^k}, for $k = n+1$, we have $E^{(n+1)} \subseteq \{x \in \B^n \ | \ w_H(x) \geq n+1\}$. But $\{x \in \B^n \ | \ w_H(x) \geq n+1\} = \varnothing$. Hence $E^{(n+1)} \subseteq \varnothing $ and $E^{(n+1)} = \varnothing$.
\end{proof}
We have effectively shown that the learning process will halt after at most $n+1$ updates. We now have to prove that when it terminates, the resulting circuit is well-tuned. Let $\mathbf{TNN}$ be the final circuit and $E = qt(f)(\mathbf{TNN}\ket{\Psi})$, then by the halting condition, we have $E = \varnothing$ and we can show:
\begin{thrm}\label{th:correctness}
    We recall that $\Phi$ is the isomorphism introduced in Section \ref{sec:A->Aq} that transforms a Boolean function into the quantum circuit corresponding to its algebraic normal form. Then:
    \\
    $E = \varnothing \Rightarrow \mathbf{TNN} = \Phi(f)$
\end{thrm}
\begin{proof}
    By definition of $E$, the fact that $E = \varnothing$ means that
    $$
        \forall x \in \B^n, \ \mathbf{TNN}\ket{x}\ket{0} = \ket{x}\ket{f(x)}
    $$
    By construction, we have $\mathbf{TNN} = \prod_{u \in \B^n}{\mathbf{C}_u^{\alpha_u}}$ with $\alpha_u \in \{0,1\}$ thus $\mathbf{TNN} \in \mathcal{A}_Q$. Let $\tilde{f} = \Phi^{-1}(\mathbf{TNN})$, then by Lemma \ref{lem:phi(f)} we must have
    $$
        \forall x \in \B^n, \ \mathbf{TNN}\ket{x}\ket{0} = \ket{x}\ket{\tilde{f}(x)}
    $$
    This means that for $x \in \B^n$, we have $\tilde{f}(x) = f(x)$, thus $\tilde{f} = f$ and $\mathbf{TNN} = \Phi(f)$.
\end{proof}
So far, we have demonstrated that when presented with a superposition of all the possible inputs, the learning algorithm terminates after at most $n+1$ updates of the circuit. Additionally, when it stops the resulting circuit is correctly tuned. But these results are conditioned to our ability to identify the states corresponding to a wrong input-output relation. We thus have to detail a way to perform the operation that we have denoted $qt$.

\section{The operation $qt$}
Let $\mathbf{TNN}$ be a tunable circuit and $\ket{\Psi}$ be a state superposition of the form $\ket{\Psi} = \sum_{x \in \B^n}{a_x\ket{x}}\ket{0}$. We remind that we want to define an operator $qt$ such that for a Boolean function $f \in \B^{\B^n}$ we have:
\begin{align*}
    &qt(f)(\mathbf{TNN}\ket{\Psi}) \\
    &=\{x \in \B^n \ | \ a_x \neq 0, \ \mathbf{TNN}\ket{x}\ket{0} = \ket{x}\ket{1 \oplus f(x)}\}
\end{align*}
So let us take $f \in \B^{\B^n}$ and define $\mathbf{O}(f)$ the oracle such that:
\begin{equation*}
    \forall x \in \B^n  \begin{cases}
                            \mathbf{O}(f)\ket{x}\ket{f(x)} = \ket{x}\ket{0} \\
                            \mathbf{O}(f)\ket{x}\ket{1 \oplus f(x)} = \ket{x}\ket{1}
                        \end{cases}
\end{equation*}
$\mathbf{O}(f)$ performs a permutation on the usual computational basis, hence it is a unitary operation. Let $\ket{\Psi'} = \mathbf{TNN}\ket{\Psi}$, then
$$
    \mathbf{O}(f)\ket{\Psi'} = W(\ket{\Psi'}) + R(\ket{\Psi'})
$$
Where we have:
$$
    W(\ket{\Psi'}) = \sum_{\substack{x \in \B^n \\ \mathbf{TNN}\ket{x}\ket{0}=\ket{x}\ket{1 \oplus f(x)}}}{a_x\ket{x}}\ket{1}
$$ 
And 
$$
    R(\ket{\Psi'}) = \sum_{\substack{x \in \B^n \\ \mathbf{TNN}\ket{x}\ket{0}=\ket{x}\ket{f(x)}}}{a_x\ket{x}}\ket{0}
$$
Measuring the read-out qubit, the probability $P_1$ of it being in state $\ket{1}$ is:
$$
    P_1 = \sum_{\substack{x \in \B^n \\ \mathbf{TNN}\ket{x}\ket{0}=\ket{x}\ket{1 \oplus f(x)}}}{|a_x|^2}
$$
The resulting circuit is given in Figure \ref{fig:qt}.
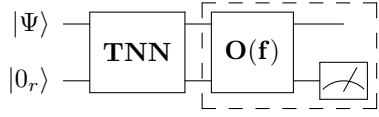
\begin{figure}[h]
    \centering
    $$
        \Qcircuit @R=1em @C=1em{
            \lstick{\ket{\Psi}}   & \multigate{1}{\mathbf{TNN}} & \multigate{1}{\mathbf{\mathbf{O}(f)}} & \qw \\
            \lstick{\ket{0_r}} & \ghost{\mathbf{TNN}}        & 
            \ghost{\mathbf{\mathbf{O}(f)}}        & \meter
            \gategroup{1}{3}{2}{4}{.7em}{--}
        }
    $$
    \caption{The operation $qt$ following the tunable neural network}
    \label{fig:qt}
\end{figure}
\\
In order to fully implement $qt$ we thus need to accurately estimate $P_1$ and reconstruct the sum resulting in this estimation. The way both of these tasks can be performed depends on the superposition we are using.

\subsection{Building a Suitable Superposition}\label{sec:superposition}
Suppose we can accurately estimate $P_1$, we want to build a set of amplitude $\{a_x \in \C \ | \ x \in \B^n\}$ such that there exist a unique subset $S \subset \B^n$ verifying:
$$
    P_1 = \sum_{x \in S}{|a_x|^2}
$$
This can be done using the uniqueness of the binary decomposition. For $x \in \B^n$ we note $x_{|10}$ its conversion in the decimal system and we consider the following superposition $\ket{\Psi}$:
$$
    \ket{\Psi} = \frac{1}{\sqrt{2^{2^n}-1}}\sum_{x \in \B^n}{\sqrt{2^{x_{|10}}}\ket{x}}\ket{0}
$$
The uniqueness of the binary decomposition yields:
\begin{align*}
    &\forall S,S' \subset \B^n, \ S \neq S'\\
    &\iff \frac{1}{2^{2^n}-1}\sum_{x \in S}{2^{x_{|10}}} \neq 
    \frac{1}{2^{2^n}-1}\sum_{x \in S'}{2^{x_{|10}}}
\end{align*}
Which is what we aimed for. Nonetheless we remind that when tuning, the gates controlled by the least number of qubits are tuned to their definitive value early in the process. We want to reflect this particular behaviour in the superposition by granting a large amplitude to inputs with a small Hamming weight. We can then switch to a superposition where inputs with a large Hamming weight have a large amplitude. For $h \in \{0,\ldots,n\}$ we define $W_h=\{x \in \B^n \ | \ w_H(x) = h \}$ For $x \in W_h$, we introduce $o_h(x)$ its rank in $W_h$. This rank is induced by the natural order $<$ on $\N$ applied to $x_{|10}$ for $x \in W_h$. We illustrate this set of functions for $n=4$ with Figure \ref{fig:o_h}.
\begin{figure}[h]
    \centering
    \begin{subfigure}{0.19\textwidth}
        \centering
        \begin{tabular}{|c|c|}
            \hline
            x    & $o_0(x)$ \\
            \hline
            0000 & 1      \\
            \hline
        \end{tabular}
        \caption{Values of $o_0$}
    \end{subfigure}
    \begin{subfigure}{0.19\textwidth}
        \centering
        \begin{tabular}{|c|c|}
            \hline
            x    & $o_1(x)$ \\
            \hline
            0001 & 1      \\
            \hline
            0010 & 2      \\
            \hline
            0100 & 3      \\
            \hline
            1000 & 4      \\
            \hline
        \end{tabular}
        \caption{Values of $o_1$}
    \end{subfigure}
    \begin{subfigure}{0.19\textwidth}
        \centering
        \begin{tabular}{|c|c|}
            \hline
            x    & $o_2(x)$ \\
            \hline
            0011 & 1      \\
            \hline
            0101 & 2      \\
            \hline
            0110 & 3      \\
            \hline
            1001 & 4      \\
            \hline
            1010 & 5      \\
            \hline
            1100 & 6      \\
            \hline
        \end{tabular}
        \caption{Values of $o_2$}
    \end{subfigure}
    \begin{subfigure}{0.19\textwidth}
        \centering
        \begin{tabular}{|c|c|}
            \hline
            x    & $o_3(x)$ \\
            \hline
            0111 & 1      \\
            \hline
            1011 & 2      \\
            \hline
            1101 & 3      \\
            \hline
            1110 & 4      \\
            \hline
        \end{tabular}
        \caption{Values of $o_3$}
    \end{subfigure}
    \begin{subfigure}{0.19\textwidth}
        \centering
        \begin{tabular}{|c|c|}
            \hline
            x    & $o_4(x)$ \\
            \hline
            1111 & 1      \\
            \hline
        \end{tabular}
        \caption{Values of $o_4$}
    \end{subfigure}
    \caption{The set of functions $o_h$ for $n=4$ and their respective values}
    \label{fig:o_h}
\end{figure}
\\
This allows us to recursively construct the following function:
\begin{align}\label{eq:p}
    \begin{cases}
        p(0) = 0 \\
        p(x) = \left(\max_{y \in W_{w_H(x)-1}}p(y)\right)+o_{w_H(x)}(x)
    \end{cases}
\end{align}
For $n=4$, we have computed the values of $p$ in Table \ref{tab:p4}.
\begin{table}[h]
    \centering
    \begin{tabular}{|c|c||c|c|}
        \hline
        $x$ & $p(x)$ & $x$ & $p(x)$ \\
        \hline
        0000 & 0  & 1000 & 4   \\
        0001 & 1  & 1001 & 8   \\
        0010 & 2  & 1010 & 9   \\
        0011 & 5  & 1011 & 12  \\
        0100 & 3  & 1100 & 10  \\
        0101 & 6  & 1101 & 13  \\
        0110 & 7  & 1110 & 14  \\
        0111 & 11 & 1111 & 15 \\
        \hline
    \end{tabular}
    \caption{Values of p(x) for $n=4$}
    \label{tab:p4}
\end{table}

We then use this function to define the following superpositions:
$$
    \ket{\Psi_{\downarrow}} = \frac{1}{\sqrt{2^{2^n}-1}}\sum_{x \in \B^n}{\sqrt{2^{2^n-1-p(x)}}\ket{x}}\ket{0}
$$
And similarly:
$$
    \ket{\Psi_{\uparrow}} = \frac{1}{\sqrt{2^{2^n}-1}}\sum_{x \in \B^n}{\sqrt{2^{p(x)}}\ket{x}}\ket{0}
$$
The  reason we introduce these superposition is that it is impossible to exactly determine $P_1$. We thus want to minimize the impact of any estimation errors. By working with $\ket{\Psi_{\downarrow}}$ first, we are more likely to correctly identify the inputs with small Hamming weight for which the output by the circuit is not that of the target function. By the update rule and Lemma \ref{lem:k_update}, we can be confident that, despite the estimation incertitude, the gates controlled by a small number of qubits can still be tuned to their definitive value early in the training process. Once these gates have been correctly tuned, we can switch to $\ket{\Psi_{\uparrow}}$ in order to tune the gates controlled by a large number of qubits.

\subsection{Estimating $P_1$}
In this section we discuss a way to estimate $P_1$ accurately enough to identify the inputs that have been misclassified by the network. For now, we consider the superposition $\ket{\Psi_{\downarrow}}$. Let $f \in \B^{\B^n}$ be the target function, $\mathbf{TNN}$ the circuit being tuned and $\mathbf{O}(f)$ the oracle we introduced earlier. Then a measurement of the read-out qubit of $\ket{\Psi'}=\mathbf{O}(f)\mathbf{TNN}\ket{\Psi_{\downarrow}}$ can be modelled by a Bernoulli process where the probability of measuring $\ket{1}$ is $P_1$. We want to determine the number of samples $s$ needed in order to estimate $P_1$ within a margin of error $\epsilon$.
\\
Let $X$ be the random variable representing the outcome of the measurement, then $X$ has a Bernoulli distribution and we have $P(X=1)=P_1$. Suppose that the number of samples is large enough, then a 95\% confidence interval CI \cite{Wallis2013} for $P_1$ is given by:
$$
    \text{CI} = \left[\overline{P_1}-z_{0.025}\frac{\sigma}{\sqrt{s}}, \ \overline{P_1}+z_{0.025}\frac{\sigma}{\sqrt{s}}\right]
$$
Where $\sigma = \sqrt{\overline{P_1}(1-\overline{P_1})}. $This means that we want:
$$
    \epsilon = z_{0.025}\sqrt{\frac{\overline{P_1}(1-\overline{P_1})}{s}}
$$
Or 
$$
    s = \frac{\sigma}{\sqrt{s}} = \frac{\overline{P_1}(1-\overline{P_1})}{\epsilon^2z_{0.025}^2}
$$
Given that $z_{0.025} = 1.96$ we can approximate with $z_{0.025} = 2$. And the term $\overline{P_1}(1-\overline{P_1})$ taking its maximum value when $\overline{P_1}=\frac{1}{2}$ we can write:
$$
    s = \frac{1}{16\epsilon^2}
$$
To determine $\epsilon$, remind that during the learning process we use two different superpositions $\ket{\Psi_{\downarrow}}$ and $\ket{\Psi_{\uparrow}}$ where the amplitude is larger for inputs with small and large Hamming weight respectively. This means that it suffices to determine $P_1$ up to a margin equal to:
$$
    \epsilon = \frac{2^{2^{n-1}}}{2^{2^n}-1} \approx \frac{1}{2^{2^{n-1}}}
$$
This in turn yields
$$
    s = O(2^{2^n})
$$
While the number of samples needed to obtain the required accuracy can seem quite large, this scale can be explained by the fact that the operation performed by $qt$ is similar to quantum state tomography. This process aims at reconstructing a quantum state by performing measurements on copies of this state. The number of copies needed to achieve this task to a precision $\epsilon$ then scales in $O\left(\frac{1}{\epsilon^2}\right)$ \cite{Haah2017}.

\subsection{Retrieving $E$}
Suppose we have $s$ samples as specified previously, and among this samples, the read-out qubit has been measured in the state $\ket{1}$ $N_1$ times, then $\overline{P_1} = \frac{N_1}{s}$. We recall that $E = \{x \in \B^n \ | \ \mathbf{O}(f)\mathbf{TNN}\ket{x}\ket{0}=\ket{x}\ket{1}\}$ which is the set of inputs for which the output computed by the tunable circuit $\mathbf{TNN}$ is wrong. We then have:
$$
    P_1 = \frac{1}{2^{2^n}-1}\sum_{x \in E}{2^{2^n-1-p(x)}}
$$
Thus:
$$
    \sum_{x \in E}{2^{2^n-1-p(x)}} \approx \left\lfloor\frac{N_1(2^{2^n}-1)}{s}\right\rfloor
$$
The uniqueness of the binary decomposition then allows us to retrieve $E$ as required.

\section{Preparation of the states $\ket{\Psi_{\downarrow}}$ and $\ket{\Psi_{\uparrow}}$}
In this section we present a way to prepare the state $\ket{\Psi_{\downarrow}}$. We remind that
$$
    \ket{\Psi_{\downarrow}} = \frac{1}{\sqrt{2^{2^n}-1}}\sum_{x \in \B^n}{\sqrt{2^{2^n-1-p(x)}}\ket{x}}\ket{0}
$$
Where $p \colon \B^n \to \N$ is the ranking function introduced in Section \ref{sec:superposition}. The last qubit being the read-out we will only focus on the first $n$ qubits. As usual, we start with the state 
$$
    \ket{\phi_0}=\ket{0}^{\otimes n}
$$
In order to get the superposition of interest, we will proceed in two steps. First generate the amplitudes, then perform a permutation to match each amplitude with its corresponding state.
\subsection{Generation of the amplitudes}
 To generate the correct amplitudes we can use the rotation operator $\mathbf{R}_y(\theta)$ defined by:
$$
    \mathbf{R}_y(\theta)= 
                        \begin{pmatrix}
                            \cos(\frac{\theta}{2}) & -\sin(\frac{\theta}{2}) \\
                            \sin(\frac{\theta}{2}) & \cos(\frac{\theta}{2})
                        \end{pmatrix}
$$
More precisely, let $\theta_n = \arccos\left(\sqrt{\frac{2^{2^{n}}}{2^{2^{n}}+1}}\right)$, we define the operator  $\mathbf{U_{\downarrow}}$ to be: 
$$
    \mathbf{U_{\downarrow}} = \mathbf{R}_y\left(2\theta_{n-1}\right) \otimes \mathbf{R}_y\left(2\theta_{n-2}\right) \otimes \ldots \otimes \mathbf{R}_y\left(2\theta_0\right)
$$
Then if $u_{1,k}$ is the element of $\mathbf{U_{\downarrow}}$ located in the first column and $k$-th row, we have $u_{1,k}=\sqrt{\frac{2^{2^n-k}}{2^{2^n}-1}}$. For example, with two qubits, we get the following operator:
\begin{align*}
    \mathbf{U_{\downarrow}} &=   \frac{1}{\sqrt{5}}
                            \begin{pmatrix}
                                2 & -1 \\
                                1 & 2
                            \end{pmatrix}
                            \otimes
                            \frac{1}{\sqrt{3}}
                            \begin{pmatrix}
                                \sqrt{2} & -1 \\
                                1 & \sqrt{2}
                            \end{pmatrix}\\
                        &=   \frac{1}{\sqrt{15}}
                            \begin{pmatrix}
                                \sqrt{8} & -2 & -\sqrt{2} & 1 \\
                                2 & \sqrt{8} & -1 & -\sqrt{2} \\
                                \sqrt{2} & -1 & \sqrt{8} & -2 \\
                                1 & \sqrt{2} & 2 & \sqrt{8}
                            \end{pmatrix}
\end{align*}
Thus applying $\mathbf{U_{\downarrow}}$ to $\ket{\phi_0}$ yields
$$
    \ket{\phi_1} = \frac{1}{\sqrt{2^{2^n}-1}}\sum_{x \in \B^n}{\sqrt{2^{2^n-1-x_{|10}}}\ket{x}}
$$
Where $x_{|10}$ is the conversion in the decimal system of $x \in \B^n$. So with $n$ 1-qubit gates, we have generated a superposition that is close to the one we are looking for, up to a permutation over the states.
\subsection{Permutation over the states}
To get $\ket{\Psi_{\downarrow}}$ we finally need to perform a permutation over some of the states. Let us denote $\sigma \colon \B^n \to \B^n$, the permutation that sorts $\B^n$ following the ranking introduced by $p$ as defined by (\ref{eq:p}), that is $\sigma(x) = p^{-1}(x_{|10})$. 

Let us first illustrate the process of building the permutation operator with a simple example. 
\begin{ex}
    For three qubits, we have:
    $$
        \sigma = 
                \begin{pmatrix}
                    000 & 001 & 010 & 011 & 100 & 101 & 110 & 111 \\
                    000 & 001 & 010 & 100 & 011 & 101 & 110 & 111
                \end{pmatrix}
    $$
    Where the second line is the image of the first line by $\sigma$. We notice that in this case, $\sigma$ is simply the transposition where 011 and 100 are swapped while the other elements stay identical: $\sigma = (011 \ 100)$. It is then quite easy to translate $\sigma$ into a quantum circuit \cite{Chuang2010}. To do so we can use a Gray code. Starting from 011 we incrementally get to 100 by flipping just one bit at a time. This yields the following sequence:
    $$
        011 \rightarrow 111 \rightarrow 101 \rightarrow 100
    $$
    This sequence then allows for a quantum circuit to be built. Each transition is represented by a $\mathbf{X}$ gate controlled by the qubits that do not flip, the target of this gate being the flipped qubit. Figure \ref{fig:perm3} depicts the permutation operator built in such a way.
    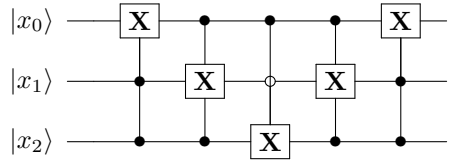
\begin{figure}[h]
        \centering
        $$
        \Qcircuit @R=1em @C=1em{
            \lstick{\ket{x_0}} & \qw & \gate{\mathbf{X}}  & \ctrl{1}           & \ctrl{2}           & 
            \ctrl{1}           & \gate{\mathbf{X}} & \qw \\
            \lstick{\ket{x_1}} & \qw & \ctrl{-1}          & \gate{\mathbf{X}}  & \ctrlo{1}          & \gate{\mathbf{X}}  & \ctrl{-1}         & \qw\\
            \lstick{\ket{x_2}} & \qw & \ctrl{-2}          & \ctrl{-1}          & \gate{\mathbf{X}}  & 
            \ctrl{-1}          & \ctrl{-2}         & \qw
        }
        $$
        \caption{Permutation operator for three qubits}
        \label{fig:perm3}
    \end{figure}
    
    To sum up the three qubits example, the circuit preparing $\ket{\Psi_{\downarrow}}$ is given in Figure \ref{fig:prep3}.
    \begin{figure}[h]
        \centering
        $$
        \Qcircuit @R=1em @C=1em{
            \lstick{\ket{0}} & \gate{\mathbf{R}_y\left(2\theta_2\right)} & \gate{\mathbf{X}}  & \ctrl{1}  & \ctrl{2}  & \ctrl{1}  & \gate{\mathbf{X}} & \qw \\
            \lstick{\ket{0}} & \gate{\mathbf{R}_y\left(2\theta_1\right)}   & 
            \ctrl{-1}          & \gate{\mathbf{X}}  & \ctrlo{1} & \gate{\mathbf{X}}  & \ctrl{-1} & \qw \\
            \lstick{\ket{0}} & \gate{\mathbf{R}_y\left(2\theta_0\right)}   & 
            \ctrl{-2}          & \ctrl{-1} & \gate{\mathbf{X}}  & \ctrl{-1} & \ctrl{-2} & \qw \\
            \lstick{\ket{0}} & \qw & \qw & \qw & \qw & \qw & \qw & \qw
        }
        $$
        \caption{Preparation of $\ket{\Psi_{\downarrow}}$ for three qubits}
        \label{fig:prep3}
    \end{figure}
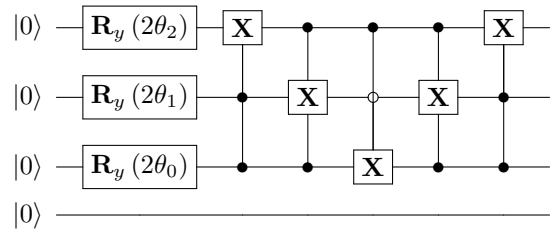
\end{ex}

In the general case, a permutation can always be decomposed as a product of transpositions. This can be done with a procedure described in Algorithm \ref{alg:permutToTransp}. The minimum used in this algorithm can be taken with respect to any order relation on $\B^n$. In our case, we will use the order relation deriving from the natural one on $\N$. 
\begin{algorithm}
    \SetAlgoLined
    prod $\leftarrow Id$\;
    $\sigma' \leftarrow \sigma$\;
    \While{$\sigma' \neq Id$}{
        $x_0 \leftarrow \min_{x \in \B^n}\{\sigma'(x) \neq x\}$\;
        prod $\leftarrow prod \circ \left(x_0 \ \sigma'(x_0)\right)$\;
        $\sigma' \leftarrow \left(x_0 \ \sigma'(x_0)\right) \circ \sigma'$\;
    }
    \textbf{return}(prod)
    \caption{Algorithm giving a decomposition of the permutation $\sigma$}
    \label{alg:permutToTransp}
\end{algorithm}
\\
To prepare the state $\ket{\Psi_{\uparrow}}$, it suffices to change the argument of the rotation operators used to generate the amplitudes by swapping the $\arccos$ for the $\arcsin$. The permutation operator is then the same as for $\ket{\Psi_{\downarrow}}$.

\section{Implementation}
The implementation of this proposition has been done using the Qiskit Python library (https://qiskit.org/) developed by IBM. While Qiskit offers the possibility to remotely run experiments on IBM quantum processors, we have limited our implementation to the quantum simulator embedded with this library. The code can be found in the following repository : https://gitlab.doc.ic.ac.uk/vjp17/tnn

\subsection{Experimental Results}
In order  to verify the theoretical results, we have implemented our quantum neural network for the cases $n=2$ and $n=3$. Most of the operations performed in our approach have been implemented using the basic set of gates provided by the library. The only gate that was not in this set is the oracle $\mathbf{O}(f)$. This gate has thus been implemented using the Operator class that allows us to create a custom gate by specifying its matrix in the computational basis.

Each function in $\BB[2]$ has been indexed according to its output. For example, we have $f_0(x)=0$ for $x \in \B^2$, $f_1(00)=1$ and $f_1(x)=0$ for $x \in \B^2 \setminus \{00\}$ and $f_{15}(x)=1$ for $x \in \B^2$. For each of these function we have run the training algorithm presented in this paper and counted the number of update steps required until convergence. To try and smooth out the effect of the random process, each experiment has been run 100 times and the results have been averaged. In addition, the theoretical values of the number of steps have also been computed. These values have been obtained with the state vector simulator of Qiskit. This simulator gives a direct access to the state vector and thus allows us to perform the operation $qt$ ideally, that is without any errors or approximations. In the theoretical case, the training has been performed with a state preparation that simply consists of Hadamard gates that generate a uniform superposition of all the possible inputs and the $qt$ operation is realised by reading the state vector output by the oracle $\mathbf{O}(f)$.

Once the training training was over, the error rate has been computed for both the experimental and theoretical cases.

The results of the experiments realised for $n=2$ are gathered in Figure \ref{fig:steps_2}. The same experiments have also been realised for $n=3$ and the results are represented in Figure \ref{fig:steps_3}.

\subsection{Analysis}
From figures \ref{fig:steps_2} and \ref{fig:steps_3} we observe that for most functions, the experimental number of steps is higher than the theoretical one. This can be explained by the errors and approximations arising from the way we perform the operation $qt$. These errors and approximations are then corrected in later steps that add to the total number of steps. We also notice some functions for which the number of experimental steps is lower than the theoretical one. In the case of these functions, some of the inputs that might be erroneous have a low amplitude and thus might not be picked by the $qt$ operation. Interestingly enough, we also notice than in the case $n=2$ and $n=3$ the number of theoretical steps is no larger than 2 which is lower than predicted. To further investigate this result, we have performed theoretical training for functions of inputs of length up to $n=10$ and observed that the training also stopped after at most 2 updates.

As can be expected, the theoretical error rate is always zero. While the experimental error rate is not null it seems to be decreasing with $n$. So even when the training does not end with a perfectly tuned network, the resulting network is still a good approximation of the target function.

\section{Conclusion}
In this paper we have presented a method that allows to compute any Boolean function by building a corresponding quantum circuit simply made out of multi-controlled NOT gates. We have also provided a formal proof of the existence as well as the uniqueness of this correspondence. From this followed our approach to quantum neural network wherein the multi-controlled gates are the elementary components performing limited computations. Still by building a circuit, or network, made out of these gates we are able to compute any Boolean functions. Furthermore, the absence of measurement in this architecture means that it is possible to work with a superposition of inputs. This ability is leveraged by the learning algorithm we have designed. When learning with a certain superposition of all the possible inputs of length $n$, the training terminates in at most $n+1$ updates of the circuit.

Nevertheless, because of the superposition, we need to perform an operation similar to quantum state tomography at each updating step. This process requires a significant number of copies of the state outputted by the quantum neural network and thus represents a bottleneck in our approach. Though the implementation of our approach also shows that it is not optimal, it also shows that the upper bound on the number of updates is lower than expected. One way to cope with the limitation of our approach could be to restrain the states present in the superposition to the ones corresponding to inputs with a certain Hamming weight. By reducing the number of states within the superposition, the margin of error is increases and thus the number of copies needed for the tomography decreases. More specifically, when performing the $k$-th update, one could consider a superposition made of the states corresponding to inputs with Hamming weight equal to $k-1$.

Thanks to its ability to handle superpositions, this architecture could be used in a Probably Approximately Correct (PAC) learning framework \cite{Valiant1984}. In the quantum version, we are provided with a superposition of the form $\ket{\phi} = \sum_{x \in \B^n}{\sqrt{D(x)}\ket{x}\ket{c(x)}}$ where $c \colon \B^n \to \B$ is the target function or concept, belonging to a class $C$ and $D$ is an unknown distribution over $\B^n$. Let $h$ be the function computed by the neural network $\mathbf{TNN}$, then the error is defined by $\text{err}_D(c,h)=\text{P}_{x \sim D}\left(c(x)\neq h(x)\right)$. Now suppose we feed $\ket{\phi}$ to $\mathbf{TNN}$, then $\ket{\phi'}=\mathbf{TNN}\ket{\phi}$ and by noting $p_1$ the probability of measuring the read-out qubit of $\ket{\phi'}$ in the state $\ket{1}$, we get $p_1 = \text{err}_D(c,h)$. An algorithm is said to $(\epsilon,\delta)$-PAC learn $C$ when $\text{P}\left(\text{err}_D(c,h) \leq \epsilon \right) \geq 1-\delta$ for all concept $c \in C$ and distribution $D$ \cite{Arunachalam2017}. The interest then lies in finding conditions on $C$ and designing an algorithm able to $(\epsilon,\delta)$-PAC learn $C$ while requiring a number of copies of $\ket{\phi}$ as low as possible.

\nocite{*}
\bibliographystyle{unsrt}
\bibliography{references}

\begin{thebibliography}{10}

\bibitem{Li2017}
H.~Li, A.~Kadav, I.~Durdanovic, H.~Samet, and H.~P. Graf.
\newblock Pruning filters for efficient convnets.
\newblock In {\em 5th International Conference on Learning Representations,
  {ICLR} 2017, Toulon, France, April 24-26, 2017, Conference Track
  Proceedings}. OpenReview.net, 2017.

\bibitem{Zhao2016}
W.~Zhao, H.~Fu, W.~Luk, T.~Yu, S.~Wang, B.~Feng, Y.~Ma, and G.~Yang.
\newblock F-cnn: An fpga-based framework for training convolutional neural
  networks, 2016.
\newblock ID: 1.

\bibitem{Lloyd2013}
S.~Lloyd, M.~Mohseni, and P.~Rebentrost.
\newblock Quantum algorithms for supervised and unsupervised machine learning.
\newblock {\em arViv e-print}, Jul 1, 2013.

\bibitem{Schuld2014}
M.~Schuld, I.~Sinayskiy, and F.~Petruccione.
\newblock The quest for a quantum neural network.
\newblock {\em Quantum Information Processing}, 13(11):2567--2586, Nov 2014.

\bibitem{Rebentrost2018}
P.~Rebentrost, T.~R. Bromley, C.~Weedbrook, and S.~Lloyd.
\newblock Quantum hopfield neural network.
\newblock {\em Physical Review A}, 98(4), Oct 2018.

\bibitem{Tacchino2019}
F.~Tacchino, C.~Macchiavello, D.~Gerace, and D.~Bajoni.
\newblock An artificial neuron implemented on an actual quantum processor.
\newblock {\em npj Quantum Information}, 5(1):1--8, Dec 2019.

\bibitem{Younes2004}
A.~Younes and J.~F. Miller.
\newblock Representation of boolean quantum circuits as reed-muller expansions.
\newblock {\em International Journal of Electronics}, 91(7):431--444, Jul 1,
  2004.

\bibitem{Bogdanov2019}
Y.~I. Bogdanov, N.~A. Bogdanova, D.~V. Fastovets, and V.~F. Lukichev.
\newblock Representation of boolean functions in terms of quantum computation.
\newblock In Vladimir~F. Lukichev and Konstantin~V. Rudenko, editors, {\em
  International Conference on Micro- and Nano-Electronics 2018}, volume 11022,
  page 740. International Society for Optics and Photonics; SPIE, 2019.

\bibitem{Younes2003}
A.~Younes and J.~Miller.
\newblock Automated method for building cnot based quantum circuits for boolean
  functions.
\newblock {\em arXiv e-prints}, pages quant--ph/0304099, Apr 2003.
\newblock quant-ph/0304099; Provided by the SAO/NASA Astrophysics Data System.

\bibitem{Kurgalin2018}
S.~Kurgalin and S.~Borzunov.
\newblock {\em The Discrete Math Workbook}.
\newblock Springer International Publishing, Cham, 2018.

\bibitem{Wallis2013}
S.~Wallis.
\newblock Binomial confidence intervals and contingency tests: Mathematical
  fundamentals and the evaluation of alternative methods.
\newblock {\em Journal of Quantitative Linguistics}, 20(3):178--208, Aug 1,
  2013.

\bibitem{Haah2017}
J.~Haah, A.~W. Harrow, Z.~Ji, X.~Wu, and N.~Yu.
\newblock Sample-optimal tomography of quantum states.
\newblock {\em IEEE Transactions on Information Theory}, 63(9):5628--5641, Sep.
  2017.

\bibitem{Chuang2010}
I.~L. Chuang and M.~A. Nielsen.
\newblock {\em Quantum Computation and Quantum Information: 10th Anniversary
  Edition}.
\newblock Cambridge University Press, Dec 9, 2010.

\bibitem{Valiant1984}
L.~Valiant.
\newblock A theory of the learnable.
\newblock {\em Communications of the ACM}, 27(11):1134--1142, Nov 5, 1984.

\bibitem{Arunachalam2017}
S.~Arunachalam and R.~de~Wolf.
\newblock Guest column: A survey of quantum learning theory.
\newblock {\em ACM SIGACT News}, 48(2):41--67, Jun 12, 2017.

\bibitem{Schmidhuber2015}
J.~Schmidhuber.
\newblock Deep learning in neural networks: An overview, 2015.
\newblock ID: 271125.

\bibitem{Iwama2002}
K.~Iwama, Y.~Kambayashi, and S.~Yamashita.
\newblock Transformation rules for designing cnot-based quantum circuits.
\newblock In {\em Proceedings 2002 Design Automation Conference (IEEE Cat.
  No.02CH37324)}, pages 419--424, June 2002.

\bibitem{Brassard2002}
G.~Brassard, P.~Hoyer, M.~Mosca, and A.~Tapp.
\newblock {\em Quantum Amplitude Amplification and Estimation}, volume 305 of
  {\em Quantum computation and information}, pages 53--74.
\newblock Amer. Math. Soc., Providence, RI, 2002.
\newblock 81P68 (81-02); 1947332.

\bibitem{Barenco1995}
A.~Barenco, C.~H. Bennett, R.~Cleve, D.~P. DiVincenzo, N.~Margolus, P.~Shor,
  T.~Sleator, J.~A. Smolin, and H.~Weinfurter.
\newblock Elementary gates for quantum computation.
\newblock {\em Physical review. A, Atomic, molecular, and optical physics},
  52(5):3457--3467, Nov 1, 1995.

\bibitem{Savage2002}
J.~E. Savage.
\newblock {\em Models of Computation}.
\newblock Addison-Wesley, Reading, Mass. [u.a.], reprinted with corr. edition,
  2000.

\bibitem{Shende2003}
V.~V. Shende, A.~K. Prasad, I.~L. Markov, and J.~P. Hayes.
\newblock Synthesis of reversible logic circuits.
\newblock {\em IEEE Transactions on Computer-Aided Design of Integrated
  Circuits and Systems}, 22(6):710--722, Jun 2003.

\end{thebibliography}
\appendix

\section{Proofs}\label{ap:proof}
\begin{proof}[Proof of Theorem \ref{th:A}]
    The commutativity of $\oplus$ and the fact that for $u \in \B^n$, $m_u=m_u^{-1}$ yield 
    $$
        \mathcal{A} = \left\{\bigoplus_{u \in \B^n}{c_u m_u} \ | \ \forall u \in \B^n, \ c_u \in \{0,1\}\right\}
    $$
    Let $\{c_u\}_{u \in \B^n}$ and $\{d_u\}_{u \in \B^n}$ such that $\bigoplus_{u \in \B^n}{c_u m_u}=\bigoplus_{u \in \B^n}{d_u m_u}$, then 
    $$
        \bigoplus_{u \in \B^n}{(c_u \oplus d_u) m_u} = \mathbf{0}
    $$
    Suppose now that there exist $u \in \B^n$ such that $c_u \neq d_u$ and denote $D=\{u \in \B^n \ | \ c_u \neq d_u\}$, then: 
    $$
        \bigoplus_{u \in \B^n}{(c_u \oplus d_u) m_u} = \bigoplus_{u \in D}{m_u} =\mathbf{0}
    $$ 
    Take $u_0 \in D$ of minimal Hamming weight, that is such that $|1_{u_0}|$ is the smallest among $D$, then 
    $$
        \bigoplus_{u \in D}{m_u}(u_0) = m_{u_0}(u_0) = 1 \neq 0
    $$
    So $D = \varnothing$ which leads to $|\mathcal{A}|=2^{2^n}$, hence $\mathcal{A} = \B^{\B^n}$.
\end{proof}

\begin{proof}[Proof of Lemma \ref{lem:C_u}]
    Let $(x_0, \ldots, x_{n-1},q_r) \in \B^{n+1}$. $\mathbf{C}_u$ being the gate controlled by the qubits $\{\ket{x_i} \ | \ i \in 1_u\}$, it will only swap the state of the read-out qubit when all the controlling qubits are in the state $\ket{1}$, meaning:
    $$
        \mathbf{C}_u\ket{x_0, \ldots, x_{n-1}}\ket{q_r} = \ket{x_0, \ldots, x_{n-1}}\ket{q_r \oplus \prod_{i \in 1_u}{x_i}} = \ket{x_0, \ldots, x_{n-1}}\ket{q_r \oplus m_u(x_0, \ldots, x_{n-1})}
    $$
\end{proof}

\begin{proof}[Proof of Theorem \ref{th:A->A_Q}]
    Let $f,g \in \mathcal{A}$ such that $\Phi(f) = \Phi(g)$, then from Lemma \ref{lem:phi(f)}:
    $$
        \forall (x_0, \ldots, x_{n-1}) \in \B^n, \ \ket{x_0, \ldots, x_{n-1}}\ket{f(x_0, \ldots, x_{n-1})} = \ket{x_0, \ldots, x_{n-1}}\ket{g(x_0, \ldots, x_{n-1})}
    $$
    That is $f = g$.
    \\
    Now let $\mathbf{G} \in \mathcal{A}_Q$, then from Lemma \ref{lem:A_Q}
    $$
        \mathbf{G} = \prod_{u \in \B^n}{\mathbf{C}_{u}^{\alpha_u^G}}
    $$
    With $\alpha_u^G \in \{0,1\}$ for $u \in \B^n$. Now let $f = \bigoplus_{u \in \B^n}{\alpha_u^G m_u}$, then $f \in \mathcal{A}$ and we have 
    $$
        \Phi(f) = \mathbf{G}
    $$
So $\Phi$ is an isomorphism from $\mathcal{A}$ to $\mathcal{A}_Q$.
\end{proof}

\begin{proof}[Proof of Lemma \ref{lem:k_update}]
    We show this result by induction over $k$. As previously, $\mathbf{TNN}^{(k)} = \prod_{u \in \B^n}\mathbf{G}_u^{(k)}$ is the circuit obtained after the $k$-th update.
    \\
    \underline{$k = 1$}
    \\
    According to Lemma \ref{lem:T_x}, we have $\mathbf{TNN}^{(0)}\ket{0,\ldots,0}\ket{0} = \mathbf{G}_{0\ldots0}^{(0)}\ket{0,\ldots,0}\ket{0}$. 
    \\
    We now apply the update rule: if $(0,\ldots,0) \in E^{(0)}$ then we switch the value of $\mathbf{G}_{0\ldots0}^{(0)}$, else, we keep it the same. Let us denote $\mathbf{G}_{0\ldots0}$ the resulting gate. 
    \\
    Following the first update we have $\mathbf{TNN}^{(1)}=\left(\prod_{u \neq (0,\ldots,0)}{\mathbf{G}_u^{(1)}}\right)\mathbf{G}_{0\ldots0}$ and:
    \begin{align*}
        \mathbf{TNN}^{(1)}\ket{0,\ldots,0}\ket{0} &= \mathbf{G}_{0\ldots0}\ket{0,\ldots,0}\ket{0} \\
        &= \ket{0,\ldots,0}\ket{f(0,\ldots,0)}
    \end{align*}
    This means that $(0,\ldots,0) \notin E^{(1)}$ and by the update rule: $\mathbf{G}_{0\ldots0}^{(2)} = \mathbf{G}_{0\ldots0}$. An induction then leads to 
    $$
        \forall s \geq 1, \ \mathbf{G}_{0\ldots0}^{(s)} = \mathbf{G}_{0\ldots0}
    $$
    This proves the case $k=1$.
    \\
    Let $\mathbf{TNN}^{(k)} = \prod_{u \in \B^n}\mathbf{G}_u^{(k)}$ be the circuit resulting from the $k$-th update. According to the induction hypothesis we have:
    $$
        \mathbf{TNN}^{(k)} = \prod_{w_H(u) \geq k}{\mathbf{G}_u^{(k)}}\prod_{w_H(u) < k}{\mathbf{G}_u}
    $$
    Where $\mathbf{G}_u$ is the final value of the gate controlled by $u$.
    \\
    Let $x \in \B^n$ such that $w_H(x) = k$, Lemma \ref{lem:T_x} yields:
    \begin{align*}
        \mathbf{TNN}^{(k)}\ket{x}\ket{0} &= \mathbf{G}_x^{(k)}\prod_{w_H(u) < k}{\mathbf{G}_u}\ket{x}\ket{0}\\
        &= \mathbf{G}_x^{(k)}\ket{x}\ket{q}
    \end{align*}
    We apply the update rule: if $\mathbf{G}_x^{(k)}\ket{x}\ket{q} = \ket{x}\ket{1 \oplus f(x)}$ we switch the value of $\mathbf{G}_x^{(k)}$, else, we keep it the same. Either way, we denote $\mathbf{G}_x$ the value resulting from the $k+1$-th update. We thus have:
    \begin{align*}
        \mathbf{TNN}^{(k+1)}\ket{x}\ket{0} &= \mathbf{G}_x\prod_{w_H(u) < k}{\mathbf{G}_u}\ket{x}\ket{0}\\
        &= \mathbf{G}_x^{(k)}\ket{x}\ket{q}\\
        &= \ket{x}\ket{f(x)}
    \end{align*}
    The update rule then states that this gate will not change value during the $k+2$-th update and an induction shows that it will not change anymore.
    \\
    This being true for all $x \in \B^n$ such that $w_H(x)=k$, we have shown the induction hypothesis for $k+1$.
\end{proof}

\section{Figures}
\begin{figure}[h]
\centering
    \includegraphics[scale=0.5]{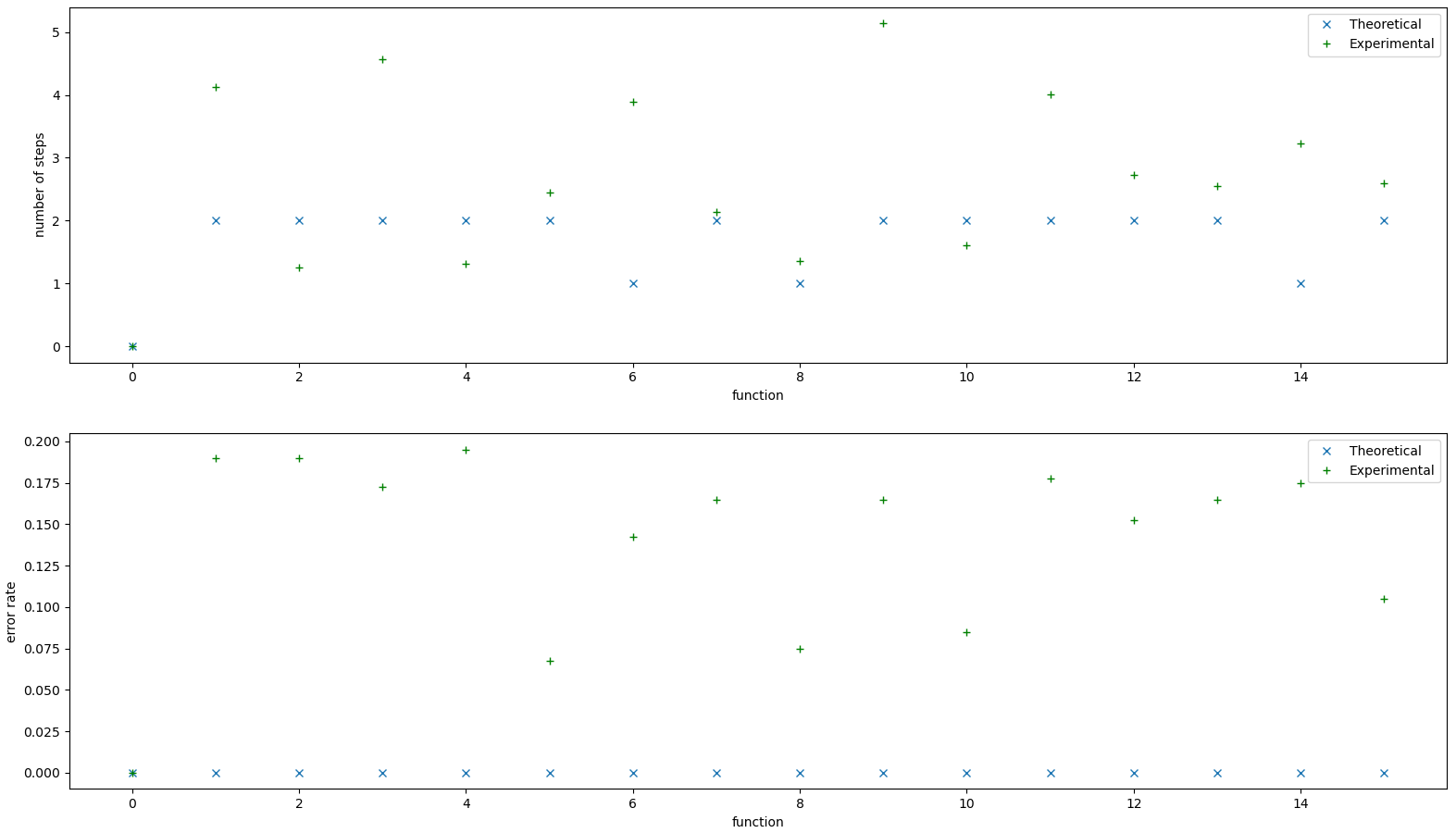}
    \caption{Number of update steps during the training phase for $n=2$}
    \label{fig:steps_2}
\end{figure}

\begin{figure}
    \centering
    \includegraphics[scale=0.5]{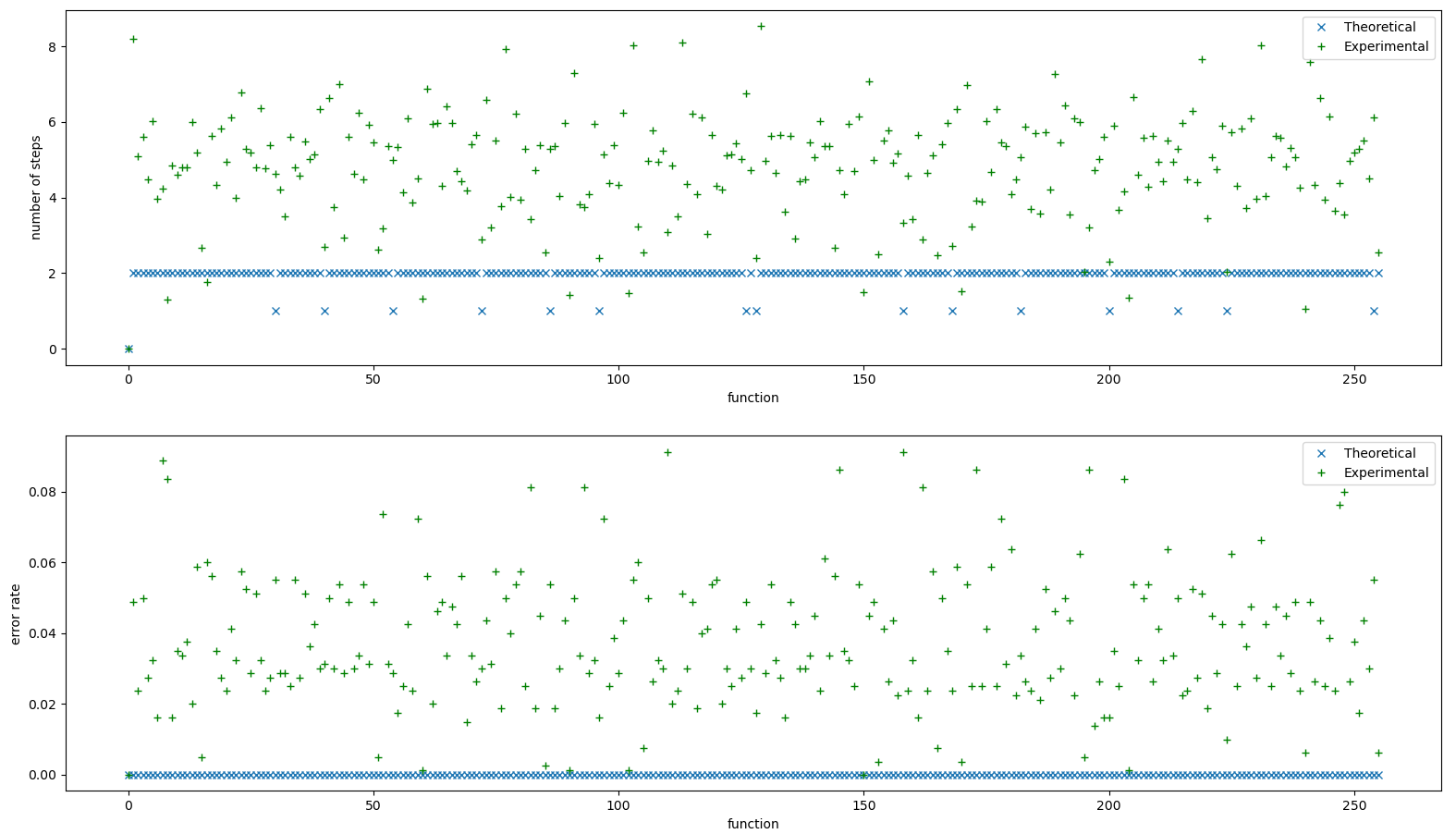}
    \caption{Number of update steps during the training phase for $n=3$}
    \label{fig:steps_3}
\end{figure}
\end{document}